\documentclass[preprints,preprints,article,accept,moreauthors,pdftex]{Definitions/mdpi}
\usepackage[frozencache]{minted}
\usepackage{amsthm}
\usepackage{subfigure}
\usepackage{mathtools}
\usepackage[ruled,vlined,linesnumbered]{algorithm2e}

\newcommand{\TV}{\mathrm{TV}}
\newcommand{\IV}{\mathrm{IV}}

\newtheorem{corollary2}{Corollary}

\newtheorem{propos}{Proposition}

\firstpage{1}
\pubvolume{xx}
\issuenum{1}
\articlenumber{1}
\pubyear{2026}
\copyrightyear{2026}
\history{Received: May 2026}

\Title{A Practical Guide to Strip Caplet Volatilities}
\Author{Fabien {Le Floc'h}\orcidA{}}
\AuthorNames{Fabien {Le Floc'h}}
\address{}
\corres{Correspondence: fabien@2ipi.com}

\abstract{We study caplet stripping, the problem of recovering a caplet volatility term structure consistent with quoted cap volatilities. Many academic papers on the Libor market model assume caplet volatilities are readily available, whereas practitioners know they are not and extracting them is a complex task. This paper presents a practical workflow, structuring the presentation around a constructive algorithm. We start with criteria on the input data based on cap time-value monotonicity. If time values fail this check, we show how to correct the quotes using robust outlier detection based on the modified Z-score. The time-value proposition naturally leads to a direct non-bootstrap stripping approach by interpolating cap time values, which yields arbitrage-free caplet volatilities by construction. We then revisit the classic sequential bootstrap approach. We introduce compact-kernel transition interpolants (flat-linear and $C^1$ flat-smooth) that preserve bootstrap equivalence. Finally, for a richer, smoother curve, we introduce global search methods using midpoint node placement with positivity-preserving calibration. Pathological cases and detailed analyses of oscillations are provided in the appendix.}


\begin{document}

\section{Introduction}
Interest rate caps and floors are widely used financial instruments that provide protection against interest rate fluctuations. A cap is a series of European call options on a floating rate (caplets), and a floor is the analogous series of put options. Caplet stripping is the process of extracting the implied volatilities of each individual caplet from market cap and floor prices, and is a standard building block for pricing, hedging, and model calibration in the interest rate derivatives market.

Caps are typically quoted in terms of implied volatility. Three conventions coexist in practice: the Black--Scholes (lognormal) volatility, which was the historical standard before negative rates became widespread; the Bachelier (normal) volatility, quoted in basis points; and the shifted lognormal volatility. Throughout this paper we work with Bachelier quotes, and stripping consists of finding the Bachelier caplet volatilities consistent with quoted cap prices. The observations we make carry over to the other conventions with minor modifications.

Let $\hat\sigma_q$ be the quoted Bachelier volatility for the cap of maturity $T_q$. The market cap price is the sum of caplet prices evaluated at the flat volatility $\hat\sigma_q$:
\begin{equation}\label{eq:cap_price}
P_q = \sum_{i=1}^{n_q} V_i(\hat\sigma_q),
\end{equation}
where $V_i(\sigma)$ is the Bachelier caplet price at vol~$\sigma$:
\begin{equation}\label{eq:caplet}
V_i(\sigma) = {B}_i^p\,\delta_i\bigl[\sigma\sqrt{t_i}\,\phi(d_i)
+ (F_i - K)\,\Phi(d_i)\bigr],
\qquad d_i = \frac{F_i - K}{\sigma\sqrt{t_i}}\,,
\end{equation}
with $\delta_i$ the year fraction, $F_i$ the forward rate, $K$ the strike, $t_i$ the time to expiry, $B_i^p$ the discount factor to payment date, $\Phi$ the standard normal cumulative distribution function, and $\phi$ its density. ${B}_i^p$ is the discount factor to payment date for the caplet with maturity $t_i$ (for a Libor rate, this is close to the discount factor to $t_{i+1}$, for an RFR such as SOFR, this is close to the discount factor to $t_i$).

As a concrete example, a standard 1Y USD cap on 3M LIBOR contains 3 caplets, with the first fixing in 3 months and settling in 6 months. The same cap on 3M SOFR (a backward-looking rate) contains 4 caplets, the first fixing and settling at roughly 3 months. Although backward-looking rates such as SOFR are increasingly prevalent, we restrict attention to forward-looking rates throughout; the modifications needed for backward-looking rates, including a linear decay of caplet volatilities during stripping \citep{piterbarg2020interest,piterbarg2020benchmark}, do not alter the main conclusions.

The goal of caplet stripping is to find individual caplet volatilities $\sigma_{i,K}$ consistent with the market prices of a set of caps at increasing maturities. Because caps are quoted at only a handful of maturities, the caplet volatility curve between quoted points is not determined by the market: it must be parameterised and interpolated. Many academic papers on the Libor market model assume caplet volatilities are readily available, whereas practitioners know they are not and extracting them is a complex task.  Different choices of interpolation and node placement can produce markedly different caplet curves, even when all of them reprice the quoted caps exactly \citep{white2014eight}.

We focus on exact or near-exact stripping. Rather than starting by emphasising what can go wrong, we organise the paper in the same order as a practical stripping workflow. We begin with criteria on the input data, checking the time values. When these criteria fail, we discuss simple repair procedures based on robust outlier detection. The time-value criterion then naturally leads to a direct stripping method based on interpolating cumulative cap time values. We next turn to the classical bootstrap approach, which remains the standard parametric alternative in practice, and finally to richer global approaches based on midpoint node placement and positivity-preserving interpolation. Pathological cases and more detailed failure analyses are deferred to the appendix.

\section{Input Data Criteria: Checking Time Values}

Before applying any stripping algorithm, one must verify the input data. There may be problematic market quotes, typically due to stale quotes.

\subsection{Arbitrage Check and Increasing Prices}
The cap prices must be increasing with the maturity of the cap, otherwise there is a calendar spread arbitrage. The difference between consecutive prices is a stream of caplets, which must have non-negative value. This translates to
\begin{equation}
P_{q+1}^{\mathrm{mkt}} \geq P_q^{\mathrm{mkt}}.
\end{equation}

\subsection{The Intrinsic--Time Value Decomposition}
A stricter criteria is to analyze the time value. Define the cumulative intrinsic value and time value at maturity $T_q$:
\begin{align}
\IV(T_q) &= \sum_{i=1}^{n_q}B_i^p\,\delta_i\,\max(F_i - K,\, 0),
\label{eq:intrinsic} \\[4pt]
\TV(T_q) &= P_q - \IV(T_q).
\label{eq:tv}
\end{align}

Each caplet's time value at vol $\sigma \geq 0$ is
\begin{equation}\label{eq:caplet_tv}
\TV_i(\sigma) = V_i(\sigma) - B_i^p\,\delta_i\max(F_i - K,\, 0).
\end{equation}
By put--call parity for each caplet,
\begin{equation}\label{eq:put_call}
V_i^{\mathrm{call}}(\sigma) = V_i^{\mathrm{put}}(\sigma)
+ B_i^p\delta_i\,(F_i - K),
\end{equation}
so
\[
\TV_i(\sigma) =
\begin{cases}
V_i^{\mathrm{put}}(\sigma) \geq 0 & \text{if } F_i \geq K
\;\;\text{(ITM caplet: TV = OTM put price)}, \\[3pt]
V_i^{\mathrm{call}}(\sigma) \geq 0 & \text{if } F_i < K
\;\;\text{(OTM caplet: TV = OTM call price)}.
\end{cases}
\]
In both cases $\TV_i(\sigma) \geq 0$ for any $\sigma \geq 0$. Thus $\TV(T_q) = \sum_{i=1}^{n_q} \TV_i(\bar\sigma_q) \geq 0$: the cumulative time value is non-negative.

\subsection{Monotonicity of Cumulative Time Value}

\begin{propos}[Incremental TV non-negativity]\label{prop:dtv}
If there exist caplet vols $\sigma_1, \ldots, \sigma_{n_q} \geq 0$ that reprice all caps simultaneously,
\[
\sum_{i=1}^{n_q} V_i(\sigma_i) = P_q, \qquad q = 1, \ldots, N,
\]
then for every $q \geq 2$,
\begin{equation}\label{eq:dtv_pos}
\Delta\TV_q \vcentcolon=
\bigl[P_q - P_{q-1}\bigr] - \bigl[\IV(T_q) - \IV(T_{q-1})\bigr]
\;\geq\; 0.
\end{equation}
\end{propos}

\begin{proof}
Under the consistent caplet vol surface, the incremental cap price is the sum of the new caplet prices:
\[
P_q - P_{q-1} = \sum_{i=n_{q-1}+1}^{n_q} V_i(\sigma_i).
\]
Each caplet price satisfies $V_i(\sigma_i) \geq B_i^p\,\delta_i\,\max(F_i - K, 0)$ (since $\TV_i \geq 0$), so
\[
P_q - P_{q-1} \geq \sum_{i=n_{q-1}+1}^{n_q} B_i^p\,\delta_i\,\max(F_i - K, 0)
= \IV(T_q) - \IV(T_{q-1}).
\]
\end{proof}

\begin{corollary2}[Arbitrage detection]\label{cor:arb}
If $\Delta\TV_q < 0$ for some $q$, then no set of non-negative caplet vols consistent with the cap quotes exists. The market quotes contain a static arbitrage.
\end{corollary2}

Table~\ref{tab:combined} illustrates an actual market dataset for on a 1M forward looking rate (Libor 1M as of February 2022) (full inputs in Appendix~\ref{app:market_data}). At 4M and 24M, the incremental time value $\Delta \TV$ is negative for $K=0$ bp. In both cases the incremental cap price falls below the incremental intrinsic: the cap vol quote implies a negative optionality premium on the new caplets.

\begin{table}[h]
  \caption{Cumulative TV and incremental $\Delta\TV$ for $K=0$ and $K=200$~bp (truncated for brevity). Negative entries indicate arbitrageable quotes.}
\label{tab:combined}
\centering
\small
\setlength{\tabcolsep}{4pt}
\begin{tabular}{rr rr>{\bfseries}r rr r}
\toprule
& & \multicolumn{3}{c}{$K = 0$~bp} & \multicolumn{3}{c}{$K = 200$~bp} \\[2pt]
\cmidrule(lr){3-5} \cmidrule(lr){6-8}
$T$ (M) & $\bar\sigma$ (bp) & $\TV$ & $\Delta\IV$ & $\Delta\TV$
& $\TV$ & $\Delta\IV$ & $\Delta\TV$ \\
\midrule
2 & 79.30 & 0.019 & 3.708 & \mdseries +0.019 & 0.000 & 0.000 & +0.000 \\
3 & 111.83 & 0.181 & 6.288 & \mdseries +0.162 & 0.004 & 0.000 & +0.004 \\
4 & 93.10 & 0.082 & 9.530 & $-$0.099 & 0.051 & 0.000 & +0.047 \\
12 & 99.86 & 1.001 & 35.898 & \mdseries +0.476 & 6.831 & 0.000 & +3.590 \\
24 & 60.42 & 0.102 & 194.588 & $-$0.899 & 20.920 & 9.827 & +14.089 \\
36 & 109.15 &  14.029 &  201.510 & \mdseries +13.926 &  110.679 &   9.258 &  +89.759 \\
60 &  96.34 &  36.868 &  366.519 & \mdseries +22.839 &  233.909 &   0.000 & +123.230 \\
\bottomrule
\end{tabular}
\end{table}

At $K=200$ bp, most caplets are out-of-the-money, and the quotes do not violate the no-arbitrage conditions. Nevertheless, stripping these quotes still yields an irregular caplet volatility curve because the underlying 3M and 24M quotes are fundamentally of poor quality. Similarly, if we apply these exact quotes to a hypothetical at-the-money scenario by artificially setting the forwards to $F_i = K = 0$ bp, the data again passes the arbitrage check but produces a distorted volatility term structure. This demonstrates that the absence of static arbitrage is a necessary but insufficient condition for generating a realistic caplet curve.
\section{Correcting Input Data}
When time values fail, we must identify which quote is the true outlier. For example, if a 4M quote fails, we do not know exactly if the problematic quote is 3M or 4M (resp. 12M or 24M). This is where the modified Z-score approach is helpful. 

\subsection{Modified Z-Score}\label{sec:mad}
Given a series of observations $x_1, \ldots, x_n$, the \emph{median absolute deviation} (MAD) is
\[
\mathrm{MAD} = \operatorname{median}\bigl(|x_i - \tilde{x}|\bigr),
\qquad \tilde{x} = \operatorname{median}(x_i).
\]
Unlike the standard deviation, the MAD has a breakdown point of~50\%: up to half the observations can be arbitrarily corrupted without affecting it. To make the MAD comparable to the standard deviation for normally distributed data, \citet{hampel1974influence} introduced the consistency factor $c = 1/\Phi^{-1}(3/4) \approx 1.4826$, so that $\hat\sigma = c \cdot \mathrm{MAD}$ is a consistent estimator of~$\sigma$.

The \emph{modified Z-score} of \citet{iglewicz1993volume} is defined as
\begin{equation}\label{eq:modz}
M_i = \frac{0.6745\,(x_i - \tilde{x})}{\mathrm{MAD}},
\end{equation}
where $0.6745 = 1/c = \Phi^{-1}(3/4)$. This normalises the score so that
$|M_i| \approx |z_i|$ when the data are Gaussian.

An observation~$x_i$ is flagged as an outlier when
\[
  |M_i| > k.
\]
\citet{iglewicz1993volume} recommend $k = 3.5$ as a conservative threshold.
The choice $k = 3$ is also widely used and corresponds (for Gaussian
data) to roughly the same false-positive rate as a $3\sigma$ rule
applied with the classical mean and standard deviation.

\subsection{Application to Cap Vol Stripping}
In the cap vol stripping context, we compute residuals of each cap vol quote relative to a local median (e.g.\ a rolling window of five neighbouring quotes):
\[
r_q = \hat{\sigma}_q - \operatorname{median}
\bigl(\hat{\sigma}_{\max(1,q-2)}, \ldots, \hat\sigma_{\min(N,q+2)}\bigr).
\]
The modified Z-score~\eqref{eq:modz} is then applied to the residual
series~$r_q$ (see Algorithm \ref{code:outlier} for a \emph{Julia} implementation).  Quotes with $|M_q| > 3$ are excluded before stripping,
and the remaining ``clean'' quotes are used to bootstrap caplet vols.
Running this logic on our dataset (Libor 1M forward looking rate as of February 2022 given in Appendix~\ref{app:market_data}), we find 2 outliers corresponding to $T_2 = 3\text{M}$ and $T_8=2\text{Y}$. 
\tablesize{\small}
\begin{table}[h]
\centering
\caption{Cap Maturities, Cap Vols (bp), and Modified Z-scores up to 7Y.}
\label{tab:cap_vols_outliers} 
\begin{tabular}{lrrrrrrrrrrrrr}
\toprule
Maturity & {2M} & \textbf{3M} & {4M} & {5M} & {6M} & {9M} & {12M} & \textbf{24M} & {36M} & {60M} & {84M}  \\
%\midrule
Cap Vol (bp) & 79.30 & \textbf{111.83} & 93.10 & 91.83 & 94.51 & 102.25 & 99.86 & \textbf{60.42} & 109.15 & 96.34 & 84.37 \\
$M_i$ & -2.49 & \textbf{3.50} & 0.00 & -0.48 & 0.00 & 1.40 & 0.00 & \textbf{-7.12} & 2.31 & 2.16 & 0.00 \\
\bottomrule
\end{tabular}
\end{table}

\section{Stripping Without Bootstrapping}
The time-value criterion naturally suggests a direct stripping approach based on interpolating cumulative cap time values. This is arguably the most immediate constructive consequence of Proposition~\ref{prop:dtv}. 

Rather than interpolating a parametric caplet vol curve and solving node values iteratively, we may convert cap vol quotes to cap prices $P_q$, interpolate those with a monotone interpolant $\hat{P}(t)$ prepending $(0, 0)$ to $(T_q, P_q)$, and recover caplet prices by differencing $V_i = \max(\hat{P}_i,\hat{P}_{i-1}, 0)$ and inverse caplet prices through a standard\footnote{For the Bachelier model, F. Le Floc'h gives a close to machine epsilon rational function approximation and thus no solver is needed. For the Black model, both F. Le Floc'h and P. J\"ackel provide a very accurate, robust and fast numerical algorithm to imply volatilities.} implied volatility solver \citep{lefloch2016fast,lefloch2026explicit, lefloch2026monotone, jackel2015let}.

Interpolated cap prices will still be increasing, but this is not enough, in particular, there is no guarantee that the difference of cap prices will be greater than the corresponding difference in intrinsic value, violating Proposition \ref{prop:dtv}. In practice, this often occurs.
The solution is to use a monotone interpolation of the \emph{time value} of the cap rather than its total price as in Algorithm \ref{alg:price_strip}. Then Proposition \ref{prop:dtv} will hold by construction and the resulting caplet volatilities will be arbitrage-free.

\begin{algorithm}[ht]
\caption{Price-based caplet vol stripping (time-value mode)}
\label{alg:price_strip}

\KwIn{
Sorted cap maturities $T_1 < \cdots < T_N$,
flat vols $\hat\sigma_1,\ldots,\hat\sigma_N$,
forwards $\{F_i\}$,
discount factors $\{B_i^p\}$,
strike $K$,
interpolator $\mathcal{I}$.
}
\KwOut{
Caplet vols $\{\sigma_i\}$ for $i=1,\ldots,n_{\max}$.
}

\BlankLine
\textbf{Step 1: Compute market cap prices}

\For{$q = 1,\ldots,N$}{
$P(T_q) \gets \sum_{i=1}^{n_q} V_i(\hat\sigma_q)$\;
}
\BlankLine
\textbf{Step 2: Arbitrage filter} (optional)

Compute $\TV(T_q) = P(T_q) - \IV(T_q)$\;
Iteratively remove quotes where $\TV(T_q) \leq \TV(T_{q-1})$ choosing the quote farther from a linear interpolation of its neighbours\;

\BlankLine
\textbf{Step 3: Build time-value nodes}

Let $\{(t_j, v_j)\}$ be the TV nodes: prepend $(0,0)$, then $(T_q,\TV(T_q))$ for remaining quotes\;

\BlankLine
\textbf{Step 4: Interpolate time value on caplet grid}

\For{$i = 1,\ldots,n_{\max}$}{
$\widehat{\TV}(t_i) \gets \mathcal{I}(\{(t_j,v_j)\}, t_i)$\;
Enforce monotonicity: $\widehat{\TV}(t_i) \gets \max\!\left(\widehat{\TV}(t_i), \widehat{\TV}(t_{i-1})\right)$
}

\BlankLine
\textbf{Step 5: Recover caplet prices}

\For{$i = 1,\ldots,n_{\max}$}{
$\Delta \TV_i \gets \max\!\left(\widehat{\TV}(t_i) - \widehat{\TV}(t_{i-1}), 0\right)$\;
$V_i \gets \Delta \TV_i + B_i^p \delta_i \max(F_i-K,0)$\;
}

\BlankLine
\textbf{Step 6: Invert for caplet vol}

\For{$i = 1,\ldots,n_{\max}$}{
$\sigma_i \gets \mathrm{BachelierImpliedVol}(V_i,T_i,\delta_i,B_i^p,F_i,K)$\;
}
\end{algorithm}

In Figure \ref{fig:price_vs_tv_interp}, we look at two monotone interpolators: the piecewise-linear interpolator and the $C^2$ cubic spline with monotonicity filter of \citet{dougherty1989nonnegativity} on the market quotes for the 1M Libor caps presented in Appendix~\ref{app:market_data}, where the outliers have been filtered out. The resulting caplet volatilities are very similar, with the monotone spline producing a smoother curve. In contrast, if we interpolate cap prices instead of time values, we get a non-monotone caplet curve with unrealistic oscillations (Figure \ref{fig:price_vs_tv_interp}). The Cap price is increasing with maturity, its derivative is always positive (Figure \ref{fig:cap_price_tv_K0_vs_K200}).
\begin{figure}[H]
\centering
\includegraphics[width=0.95\textwidth]{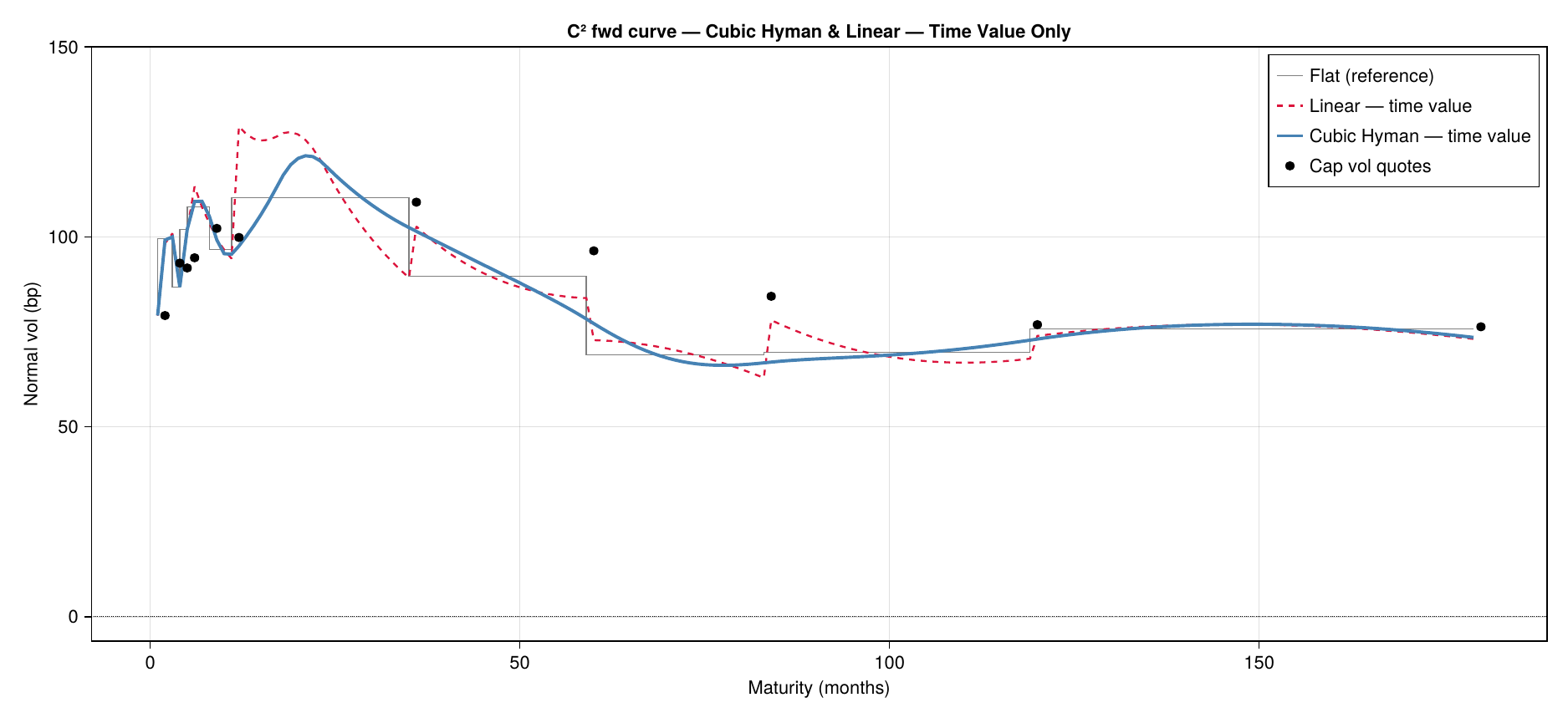}
\caption{Caplet volatility with time-value interpolation using smooth interest rate curves. \label{fig:price_vs_tv_interp}}
\end{figure}
\begin{figure}[H]
\centering
\includegraphics[width=0.95\textwidth]{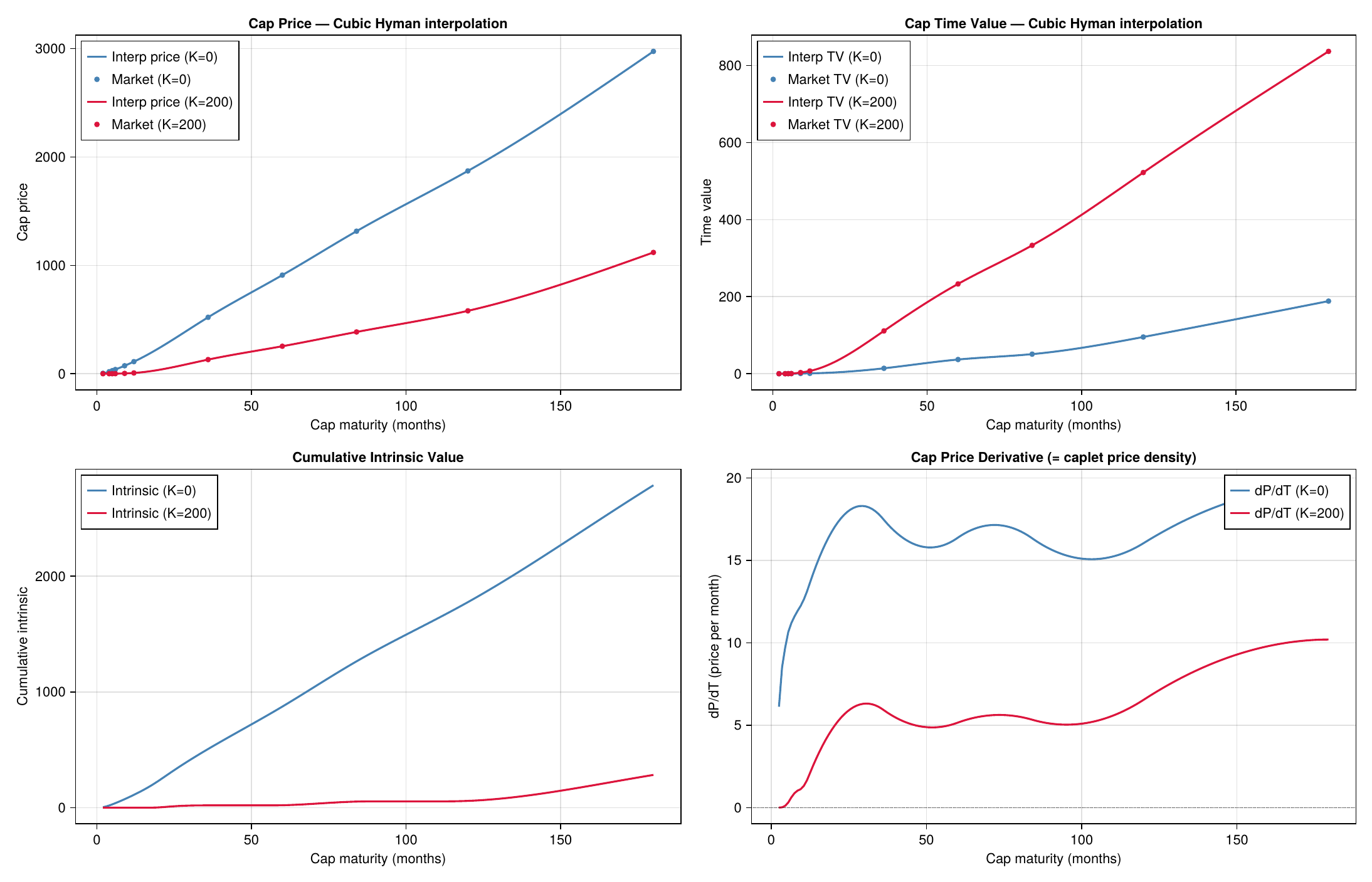}
\caption{Cap prices for different strikes, using the same arbitrage-free cap vol quotes (outliers removed). \label{fig:cap_price_tv_K0_vs_K200}}
\end{figure}
A primary drawback of this stripping approach is that the smoothness of the resulting caplet volatilities relies heavily on the smoothness of the forward curve (see Appendix \ref{sec:forward_curve} for a counterexample). To account for this in Figure \ref{fig:price_vs_tv_interp}, we deliberately applied a smooth interpolation to the forward rates rather than a flat-forward method. This may not be such a bad requirement in practice, since the cap volatilities are likely quoted assuming a smooth forward curve --- a pricing curve as opposed to a risk curve which typically uses log-linear interpolation. A risk curve would result in different forward rates and thus different caplet volatilities, regardless of the stripping technique. 

A second limitation involves handling extrapolation beyond the final cap maturity, which requires introducing artificial quotes at longer maturities. For instance, mimicking a flat caplet volatility extrapolation might involve adding quotes at every tenor to force a constant cap volatility, though this fails to remain truly flat between tenors. Alternatively, one could add a single distant quote (e.g., at 50 years) matching the last quoted cap volatility (see Appendix \ref{sec:extrapolation}); the choice of date will impact the interpolation up to the last quoted tenor. Ultimately, neither of these workarounds is entirely satisfactory.

\section{The Classic Bootstrap}
A different route, and the one historically used most often in practice, is to parameterise the caplet volatility term structure itself and solve sequentially for node values. We have $N$ cap vol quotes $(\hat\sigma_1,\dots,\hat\sigma_N)$ at maturities $T_1<\cdots<T_N$. The incremental constraint (cap~$k$ minus cap~$k{-}1$) pins down the newly added caplets, leading to a simple lower triangular system. This sequential algorithm is formalized in Algorithm \ref{alg:bootstrap_strip}. 
\begin{algorithm}[ht]
\caption{Sequential bootstrap caplet stripping}
\label{alg:bootstrap_strip}

\KwIn{
Sorted cap maturities $T_1 < \cdots < T_N$,
quoted flat cap vols $\hat{\sigma}_1,\ldots,\hat{\sigma}_N$,
forwards $\{F_i\}$,
discount factors $\{B_i^p\}$,
accrual factors $\{\delta_i\}$,
strike $K$,
caplet tenor $\Delta$,
node placement rule $\{\tau_q\}$,
interpolation rule $\sigma(\cdot; v_1,\dots,v_q)$.
}

\KwOut{
Bootstrap node values $v_1,\dots,v_N$ and stripped caplet vols $\sigma_i=\sigma(t_i)$.
}

\BlankLine
\textbf{Step 1: Compute market cap prices}

\For{$q=1,\ldots,N$}{
$P_q \gets \sum_{i=1}^{n_q} V_i(\hat{\sigma}_q)$\;
}

\BlankLine
\textbf{Step 2: Choose bootstrap nodes}

For a forward-looking rate, set $\tau_q = T_q-\Delta$;\;
For a backward-looking rate, set $\tau_q = T_q$\;
Choose an interpolation family, for example piecewise-constant, flat-linear, or flat-smooth\;

\BlankLine
\textbf{Step 3: Bootstrap sequentially}

\For{$q=1,\ldots,N$}{
Treat $v_1,\dots,v_{q-1}$ as fixed\;
Solve for $v_q$ so that the model cap price matches the market cap price:
\[
\sum_{i=1}^{n_q} V_i\bigl(\sigma(t_i; v_1,\dots,v_q)\bigr)=P_q
\]
or equivalently, for $q\ge 2$, solve the incremental equation
\[
\sum_{i=n_{q-1}+1}^{n_q} V_i\bigl(\sigma(t_i; v_1,\dots,v_q)\bigr)
=
P_q-P_{q-1}
\]
using a one-dimensional root finder\;
}

\BlankLine
\textbf{Step 4: Recover caplet volatilities}

\For{$i=1,\ldots,n_{\max}$}{
$\sigma_i \gets \sigma(t_i; v_1,\dots,v_N)$\;
}
\end{algorithm}
The parameterization of the caplet vol curve becomes independent of the interpolation method used for the interest rate curve: it will be smooth as long as the interpolation of the caplet vol curve is smooth, even if the forward curve has jumps.

Let $\Delta$ be the  caplet frequency (e.g.\ $\Delta = 1$ month), a cap on a \emph{Forward-looking} rate (e.g.\ LIBOR) of maturity~$T_k$
            contains caplets fixing at
            $\Delta,\, 2\Delta,\, \ldots,\, T_k - \Delta$.
            The last caplet fixes one period before~$T_k$ and pays at~$T_k$.
      A cap on a \emph{Backward-looking} rate (e.g.\ SOFR) of maturity~$T_k$
            contains caplets fixing at
            $\Delta,\, 2\Delta,\, \ldots,\, T_k$.
    
 Let $\tau_k$ be the node times, for forward-looking rates the node is placed at
        the last fixing time of each cap, $\tau_k = T_k - \Delta$ while for
        backward-looking rates,  $\tau_k = T_k$.
The piecewise-constant interpolation is left-continuous on $\tau_k$ and reads 
\begin{equation*}\sigma(t) = \sum_{q=1}^{M_q} \sigma_{q,K} 1_{t \in (\tau_{q-1}, \tau_q]} + \sigma_{M_q,K} 1_{t \in [\tau_{M_q}, \infty)}\,,\end{equation*}
with $\tau_0 = 0$.

A classic bootstrap with piecewise-constant interpolation leads to Figure \ref{fig:caplet_vols_flat}.
\begin{figure}[ht]
\centering
\includegraphics[width=0.95\textwidth]{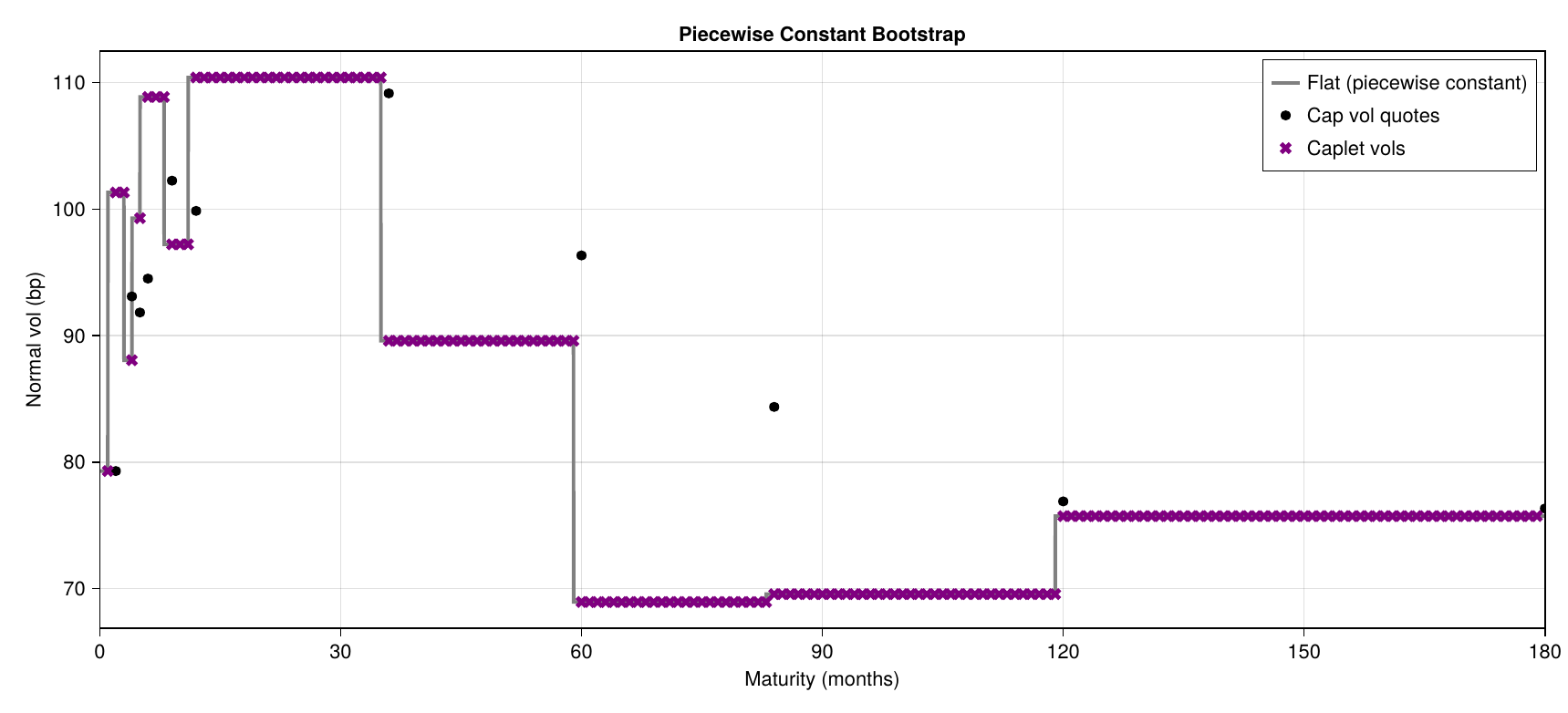}
\caption{Caplet volatilities obtained with bootstrapping and piecewise-constant interpolation. The cap vols appear shifted by 1M because a cap maturity of 3M actually has its last fixing at 2M for a forward looking rate with 1M tenor.\label{fig:caplet_vols_flat}}
\end{figure}
The standard piecewise-constant (step-forward) interpolation of the caplet
vol term structure produces sharp discontinuities at every node boundary.
While the bootstrap is exact and sequential, the resulting vol curve is
visually unappealing and can cause numerical issues in applications that
evaluate the curve at non-caplet times (e.g.\ intra-month pricing or Theta like risk).

\subsection{What would a child do?}

The piecewise-constant interpolation resulting from the standard sequential approach is left-continuous on $\tau_k$. Many libraries and quants have historically used this stepwise-constant interpolation with little real justification. Ironically, a young child, when asked to connect the dots in Figure \ref{fig:caplet_vol_scatter}, would naturally\footnote{Confirmed with my own children.} draw continuous lines between them. This childish instinct is in fact more appropriate, because a continuous transition avoids unrealistic jumps in the volatility curve and is much smoother than the quant's discontinuous leaps.
\begin{figure}[ht] 
\centering
\includegraphics[width=0.95\textwidth]{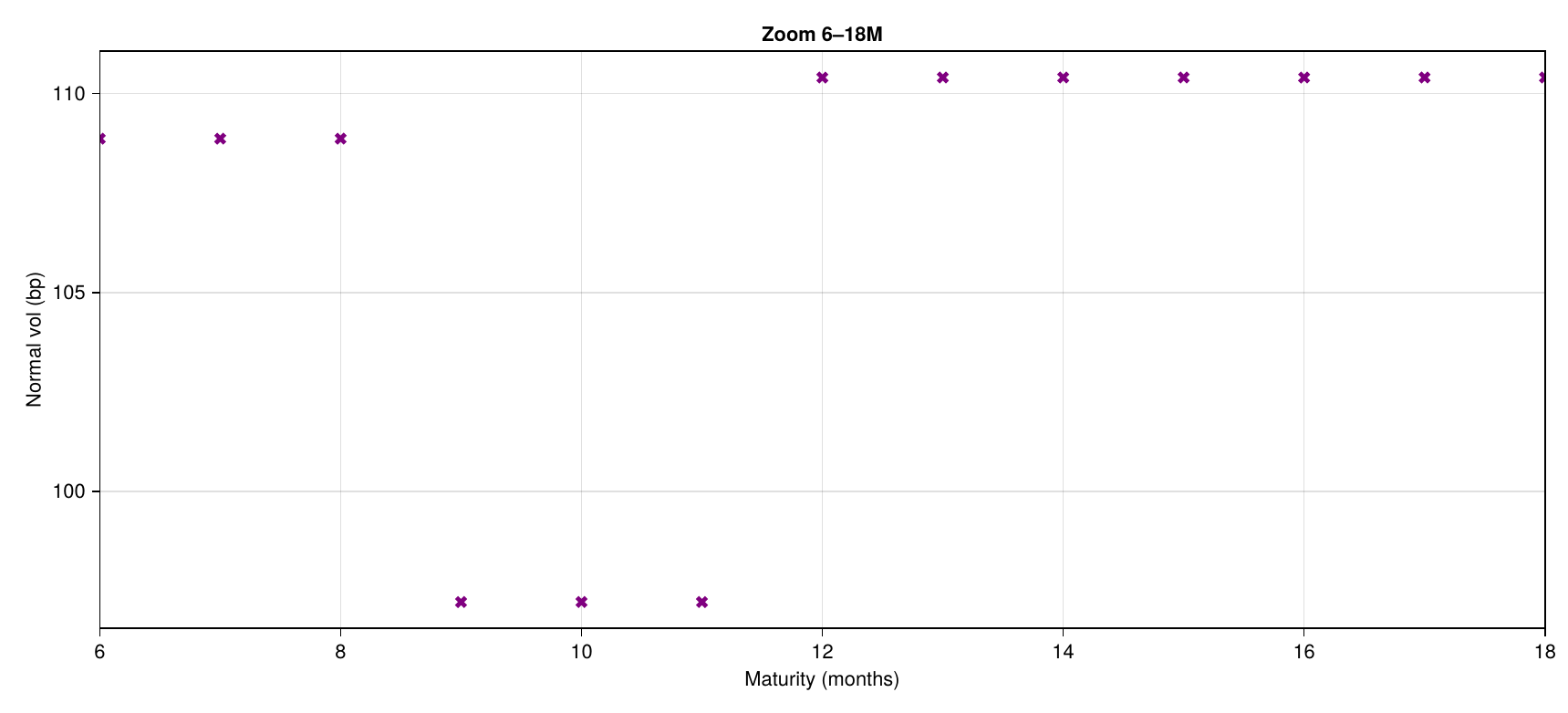}
\caption{Discrete caplet volatilities obtained from the sequential bootstrap, representing the nodes that must be interpolated.}
\label{fig:caplet_vol_scatter}
\end{figure}

We can formalise this continuous intuition by introducing a \emph{flat-linear} interpolation. We introduce a ramp factor $\beta \in [0, 1]$. Let $\bar\beta= \beta\Delta$ be the ramp width. Let $v_q = \sigma_{q,K}$ be the node value. The flat-linear interpolant $\sigma(t)$ uses a ramp centred at $c_k = \tau_{k-1} + \tfrac{\Delta}{2}$ with boundaries $a_k$ and $b_k$:
\begin{equation}\label{eq:flat-linear}
\sigma(t) =
\begin{cases}
v_{k-1},
& \tau_{k-1} < t \leq a_k, \\[4pt]
\displaystyle v_{k-1} + \frac{t - a_k}{b_k - a_k}\,
\bigl(v_k - v_{k-1}\bigr),
& a_k < t \leq b_k, \\[8pt]
v_k,
& b_k < t < \tau_k,
\end{cases}
\end{equation}
and \[
  \sigma(t) = v_1 \quad\text{for } t \leq \tau_1,
  \qquad
  \sigma(t) = v_N \quad\text{for } t \geq \tau_N.
\]
with
\[
  a_k = \max\bigl(\tau_{k-1},\; c_k - \tfrac{\bar\beta}{2}\bigr),
  \qquad
  b_k = \min\bigl(\tau_k,\; c_k + \tfrac{\bar\beta}{2}\bigr).
\]
At the nodes themselves the step-forward convention gives
$\sigma(\tau_k) = v_k$. 
Figure \ref{fig:flat_vs_flat_linear_zoom} shows the impact of the ramp width $\beta$ on the caplet volatilities.

\begin{figure}[ht]
\centering
\includegraphics[width=0.95\textwidth]{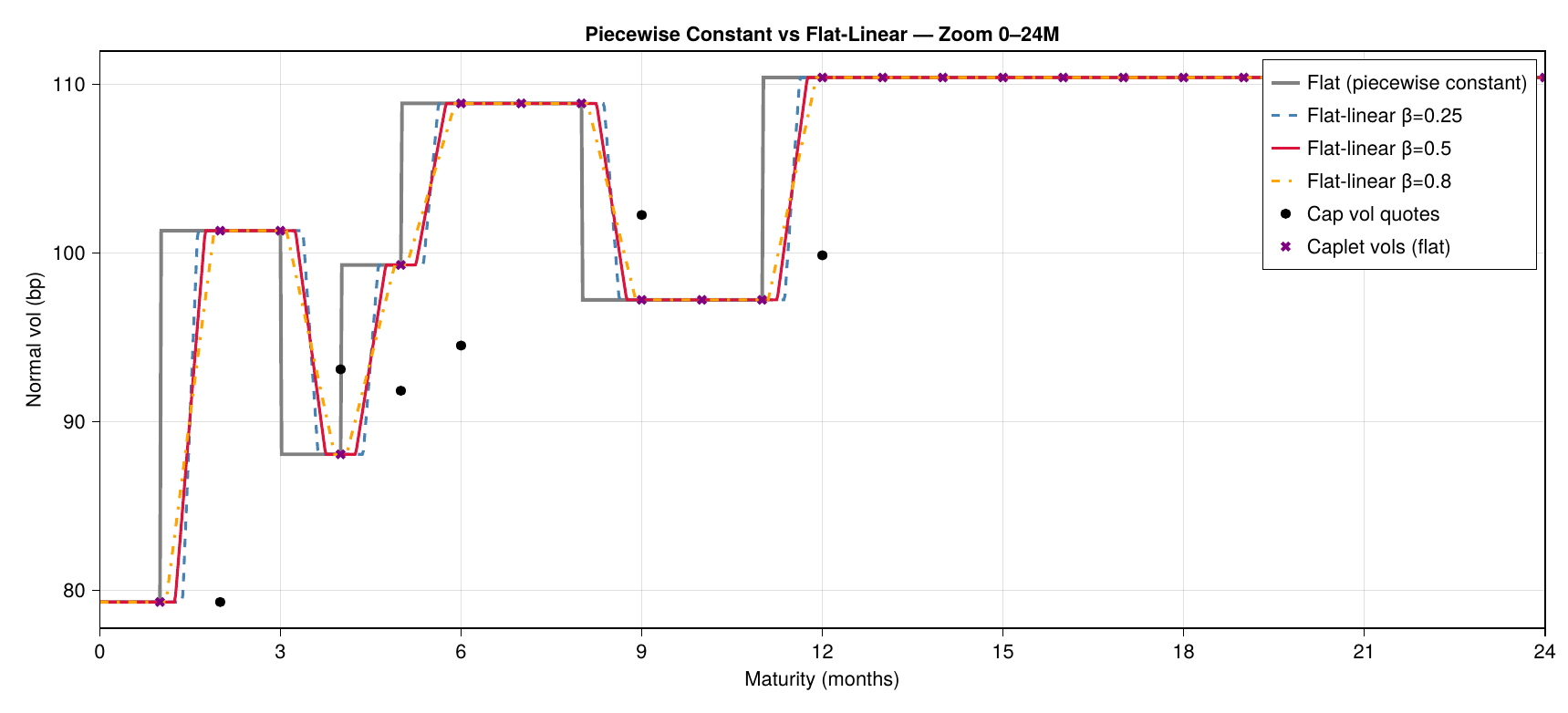}
\caption{Caplet volatilities for maturities below 24M with flat-linear interpolation and various ramp width $\beta$. \label{fig:flat_vs_flat_linear_zoom}}
\end{figure}

\subsection{What might a draughtsman do?}
The \emph{flat-smooth} interpolant replaces the linear ramp with the cubic Hermite smoothstep $\Psi(s) = 3s^2 - 2s^3$:
\begin{equation}\label{eq:flat-smooth}
\sigma(t) =
\begin{cases}
v_{k-1},
& \tau_{k-1} < t \leq a_k, \\[4pt]
\displaystyle v_{k-1} + \Psi\left(\frac{t - a_k}{b_k - a_k}\right)\,
\bigl(v_k - v_{k-1}\bigr),
& a_k < t \leq b_k, \\[8pt]
v_k,
& b_k < t < \tau_k,
\end{cases}
\end{equation}
Because $\Psi'(0) = \Psi'(1) = 0$, the resulting curve is $C^1$ (no slope discontinuities) while preserving bootstrap equivalence for $\beta \leq 1$. 

\begin{figure}[ht]
\centering
\includegraphics[width=0.95\textwidth]{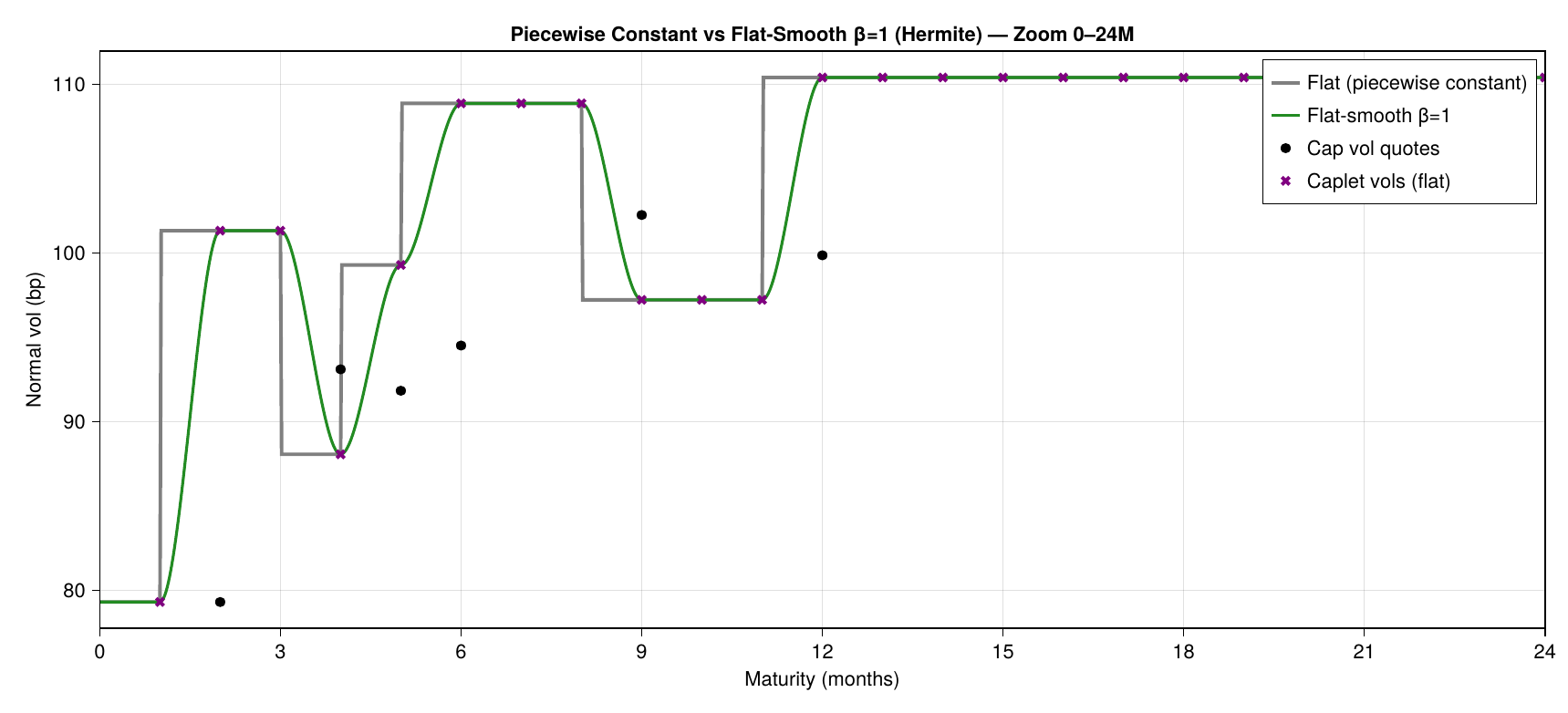}
\caption{Caplet volatilities for maturities below 24M with flat-smooth interpolation with $\beta = 1$. \label{fig:flat_vs_flat_smooth_zoom}}
\end{figure}

\begin{propos}
If $\beta \leq 1$, the flat-linear (or flat-smooth) interpolant evaluates to the same value as the piecewise-constant interpolant at every caplet fixing date. Consequently, sequential bootstrap produces identical node values.
\end{propos}
\begin{proof}
Caplet fixing times are integer multiples of~$\Delta$. The ramp centre is $c_k = \tau_{k-1} + \Delta/2$. For any caplet fixing time $t$ falling inside $(\tau_{k-1}, \tau_k)$, $t \geq b_k$ and $\sigma(t) = v_k$. At node times the step-forward convention gives $\sigma(t) = v_k$ directly.
\end{proof}

\subsection{Generalisation: Compact Kernel Transitions}
The flat-linear ramp is equivalent to convolving each step-function jump
with a rectangular kernel of width~$\bar\beta$.  This generalises to any
compact-support kernel~$K$ with
$\int_0^1 K(s)\,ds = 1$.  Define the transition function
\[
  \Psi(s) = \int_0^s K(u)\,du, \qquad s \in [0, 1],
\]
so that $\Psi(0) = 0$ and $\Psi(1) = 1$.  The interpolant becomes
\[
  \sigma(t) = v_1 + \sum_{k \geq 2} (v_k - v_{k-1})\,
    \Psi\left(\frac{t - a_k}{b_k - a_k}\right),
\]
with $\Psi$ clamped to $[0,1]$ outside $[0,1]$.

The flat-linear scheme uses $K \equiv 1$ (rectangular) giving
$\Psi(s) = s$.  The flat-smooth scheme uses the Hermite smoothstep kernel.
\tablesize{\small}
\begin{center}
\begin{tabular}{lll}
  \toprule
  \textbf{Kernel} & \textbf{Transition $\Psi(s)$} & \textbf{Smoothness} \\
  \midrule
  Rectangular ($K = 1$) & $s$ & $C^0$ (flat-linear) \\
  Hermite smoothstep & $3s^2 - 2s^3$ & $C^1$, $\Psi'(0) = \Psi'(1) = 0$ \\
  Cosine (Hann) & $\tfrac{1}{2}(1 - \cos \pi s)$ & $C^1$ \\
  Quintic smootherstep & $6s^5 - 15s^4 + 10s^3$ & $C^2$\\
  \bottomrule
\end{tabular}
\end{center}

\begin{figure}[ht]
\centering
\includegraphics[width=0.95\textwidth]{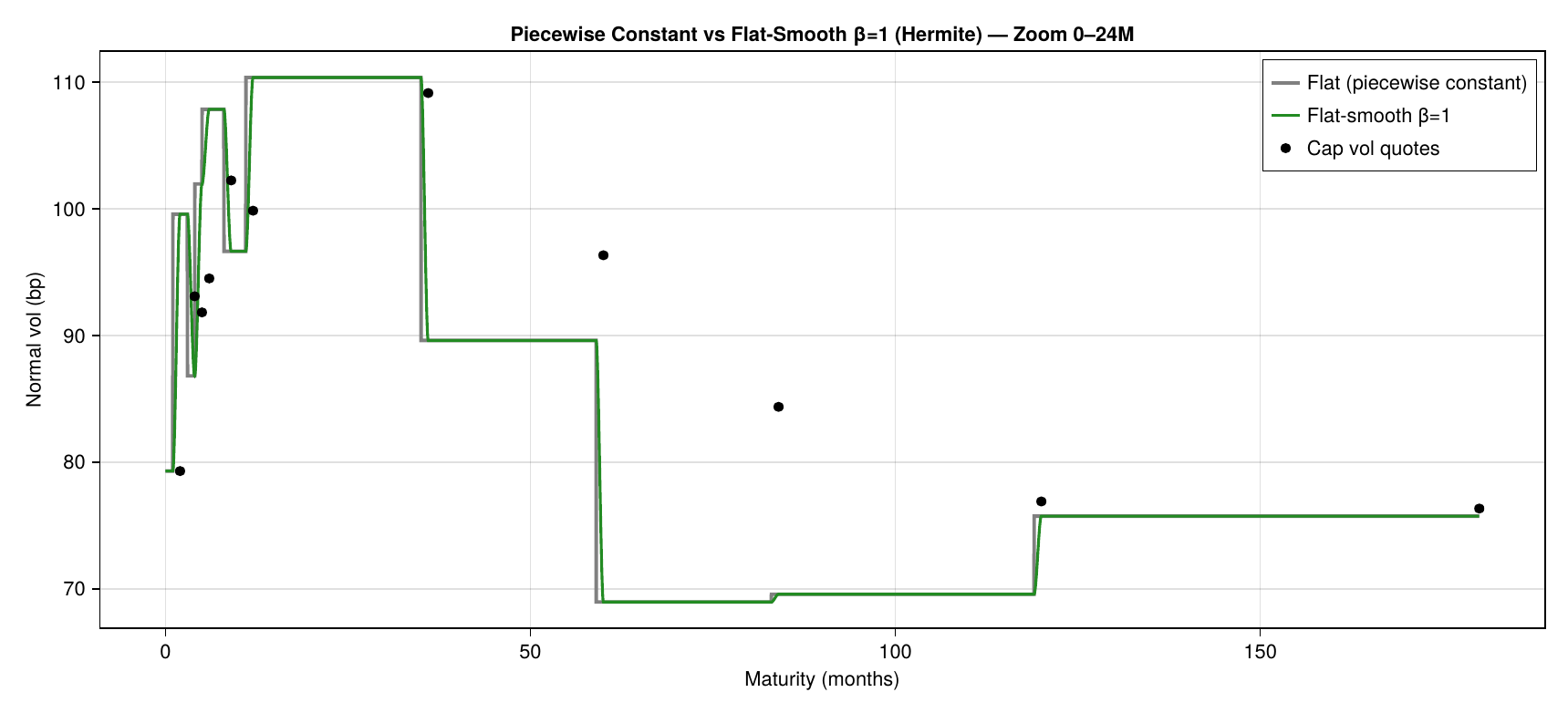}
\caption{Caplet volatilities for all maturities with flat-smooth interpolation with $\beta = 1$. \label{fig:flat_vs_flat_smooth}}
\end{figure}

\section{Smoother Interpolations}
These flat-linear and flat-smooth bootstrap methods are simple and a clear improvement, at least visually (Figures \ref{fig:flat_vs_flat_smooth_zoom} and \ref{fig:flat_vs_flat_smooth}), over the classic piecewise-constant interpolation. We propose to go further with a global curve at midpoints. The simple bootstrap can still be used to give a quick reasonable starting value.

\subsection{Midpoint Node Placement}
We warn about problems with node placement at maturity that is a common approach we have seen implemented: it may cause the solver to overshoot. 
To fix this, we instead consider two kinds of node sets:
\begin{itemize}
\item \emph{At-maturity} nodes set $\tau_k = T_k - \delta$.
\item \emph{Midpoint} nodes set $\tau_1 = T_1 - \delta$ and $\tau_k = (T_{k-1}+T_k)/2 - \delta$ for $k\ge 2$.
\end{itemize}

The introduction of midpoint nodes leads to a non-lower-triangular matrix, and thus requires a global solver to find the caplet volatilities. We use the Levenberg-Marquardt method with backtracking line search as global solver \citep{nocedal2006numerical}.

\begin{figure}[ht]
\centering
\includegraphics[width=0.95\textwidth]{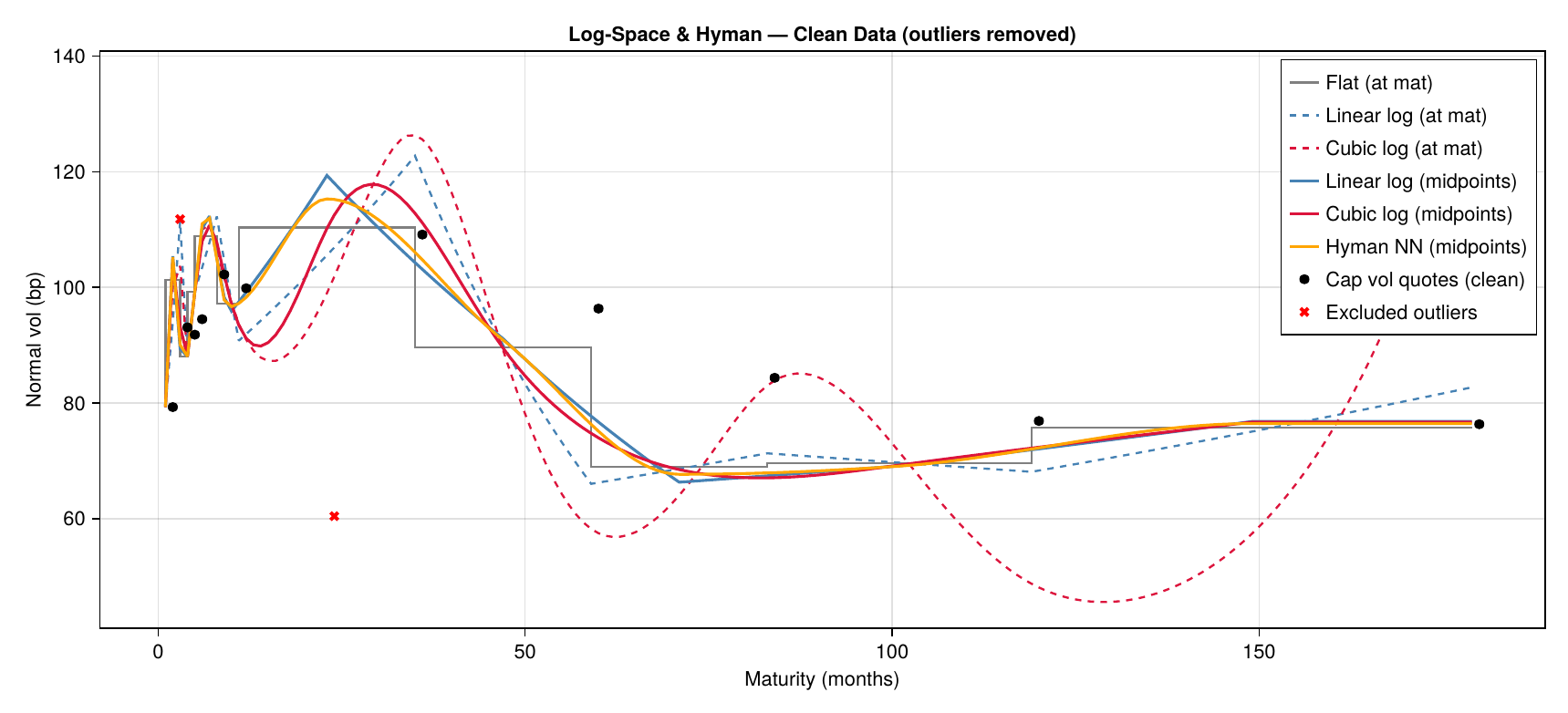}
\caption{Caplet volatilities obtained with outliers removed using various interpolation methods applied on mid-point or at-maturity nodes.}
\label{fig:clean_comparison_plot}
\end{figure}
When we apply these methods to our dataset with the previously identified outliers removed, the benefits of midpoint placement and smooth interpolation become clear. Figure \ref{fig:clean_comparison_plot} shows the resulting interpolated caplet volatilities across these different methods on the clean data. Midpoint node placement proves highly effective.

If we keep the outliers cap quotes, the effect of midpoint placement and positivity constraint is even more pronounced (Appendix \ref{sec:oscillations}). 

\subsection{Enforcing Positivity}
A priori, there is no guarantee that linear interpolation at midpoints will lead to positive caplet volatilities. We enforce positivity during optimization by using a change of variable $v_k = \exp(u_k)$ where $u_k$ is the optimization variable.

For cubic splines, \citet{dougherty1989nonnegativity} propose a method to enforce non-negativity of the spline by filtering the first derivatives. In our implementation, we use a $C^1$ cubic spline based on Bessel (parabolic) estimate of derivatives along with the non-negative and monotonicity preserving filter instead of a $C^2$ cubic spline. We pay attention to use the corrected version for the filter coefficients as described in \citep{chasethedevil2017hyman_typo} and reproduced in Appendix~\ref{app:hyman_typo}.

In Figure \ref{fig:clean_comparison_plot}, the non-negative Hyman interpolation is close to the linear midpoints but much smoother. The last mid-point is at 150M, after which the splines will revert to flat extrapolation. The boundary condition we use for the non-negative Hyman spline is thus a first derivative of zero at the last node.

\subsection{Visual Summary}
Figure \ref{fig:visual_summary} provides a visual summary of the relevant caplet stripping techniques presented so far: from the time-value interpolation using a monotonicity preserving Hyman filter on a $C^2$ cubic spline to the bootstrap method using a non-negative $C^1$ Hyman spline at mid-points.
\begin{figure}[ht]
\centering
\includegraphics[width=0.95\textwidth]{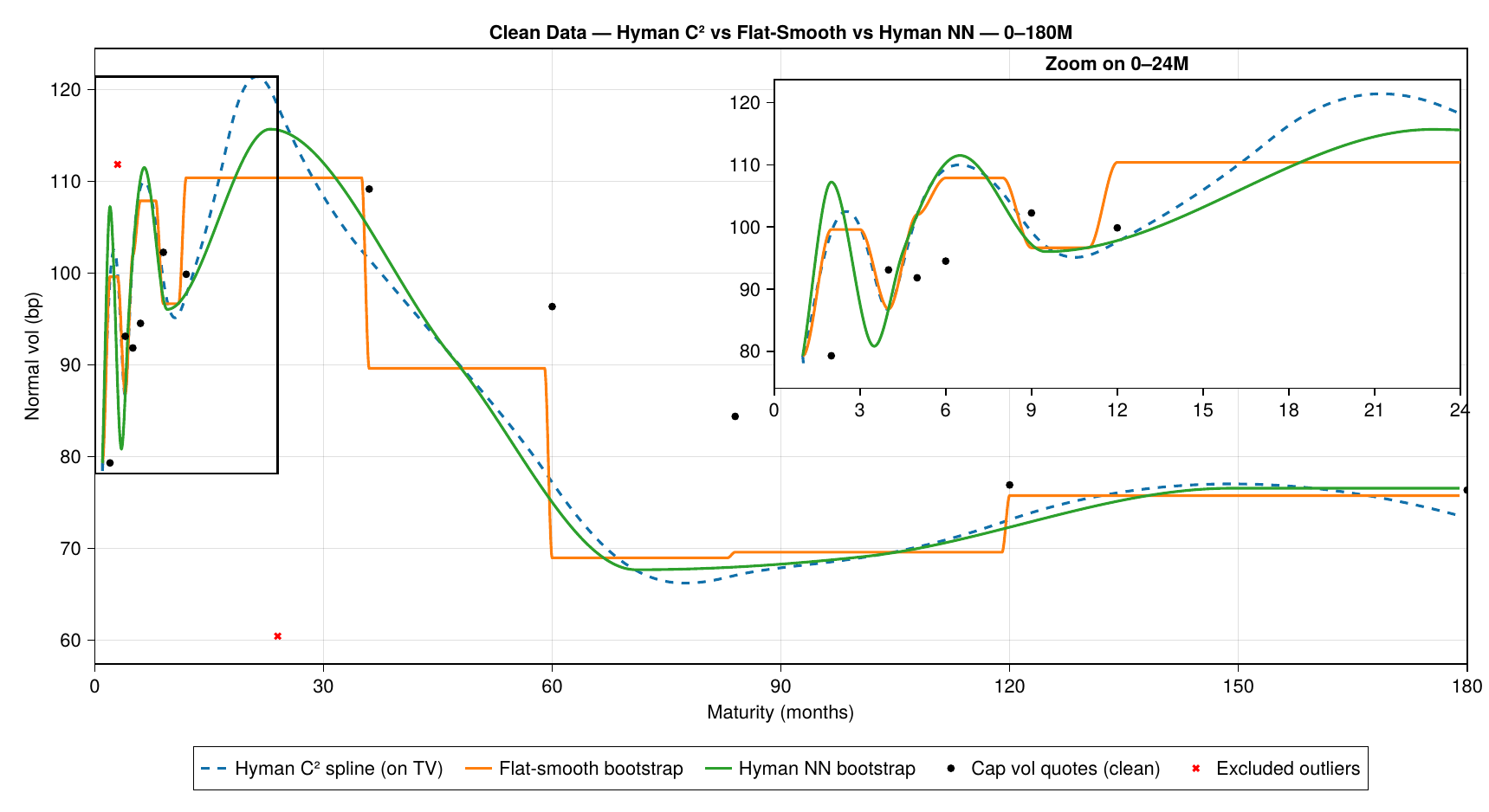}
\caption{Caplet volatilities obtained with outliers removed using various stripping methods.}
\label{fig:visual_summary}
\end{figure}

\section{Conclusion}
This paper has examined how interpolation, node placement, and quote quality affect exact or near-exact caplet stripping. The main finding is that large oscillations in stripped caplet volatilities are not an inevitable consequence of exact repricing, but arise mainly from how the caplet volatility curve is parameterised. In particular, interpolation at cap-maturity nodes can amplify local irregularities in the cap quotes and generate unrealistic caplet volatility shapes.

The paper shows that data quality must be checked before blaming the stripping method itself. Simple diagnostics based on the monotonicity of cap time values and robust outlier detection using MAD are effective at identifying problematic quotes. The paper then presented a direct non-bootstrap approach based on interpolating cap time values rather than cap prices, which enforces arbitrage-free caplet prices by design.

For bootstrap-based stripping, a pure piecewise-constant parameterisation is difficult to justify. Replacing it with the proposed flat-linear interpolation already yields a continuous caplet volatility curve while preserving the same bootstrap solution. For smoother constructions, midpoint node placement combined with a global solver is much better behaved. Among the methods tested here, the positivity-preserving midpoint approaches, especially the non-negative $C^1$ Hyman spline, provided the best compromise between smoothness, stability, and exact repricing.

\acknowledgments{The author thanks Gary Kennedy for contributing to some key ideas presented in this paper as well as for his insightful comments and suggestions on an earlier draft of this paper. }
\funding{This research received no external funding.}
\conflictsofinterest{The authors declare no conflict of interest.}
\externalbibliography{yes}
\bibliography{lefloch_caplet_stripping}
\appendixtitles{no}

\appendix

\section{Oscillations and Blow-ups}\label{sec:oscillations}
The resulting caplet volatilities obtained with linear interpolation and at-maturity nodes show strong oscillations, with a node near 18 months settling at $-159$\,bp. With midpoint nodes, the oscillations are much reduced, and the node near 18 months settles at $-14$\,bp.

\subsection{The Reason for the Oscillations}
With linear interpolation and at-maturity nodes $\tau_k = T_k$, the vol at caplet time $t_i \in (T_{k-1}, T_k]$ is
\[
\sigma(t_i) = v_{k-1}\,\frac{T_k - t_i}{T_k - T_{k-1}}
+ v_k\,\frac{t_i - T_{k-1}}{T_k - T_{k-1}}.
\]
The sequential bootstrap solves for $v_k$ given $(v_1,\dots,v_{k-1})$ already fixed. The incremental constraint involves caplets whose vols depend on $v_k$ with weight
\[
w_i = \frac{t_i - T_{k-1}}{T_k - T_{k-1}},
\qquad w_i \in \bigl[\tfrac{\Delta}{T_k-T_{k-1}},\; 1\bigr].
\]
Since $w_i$ increases linearly from $\approx 0$ to $1$, the average weight is $\bar w \approx 1/2$ which is roughly \textbf{half} the leverage of flat interpolation (where $w_i=1$). When the target incremental vol drops sharply (e.g.\ the 2Y dip), $v_k$ must \emph{overshoot} to compensate for the diluted sensitivity. In our numerical example, node values of $-175$\,bp emerge even though the target incremental vol is merely low, not negative.

\subsection{Why midpoints help}
With midpoint nodes $\tau_k = (T_{k-1}+T_k)/2$, the cap maturity $T_q$ falls near the right edge of interval $[\tau_q,\tau_{q+1}]$ rather than coinciding with a node. Crucially, the Jacobian $J_{qk} = \partial P_q/\partial v_k$ is no longer lower-triangular: each node~$v_k$ has substantial influence on at least two cap prices. The global Newton system distributes corrections bilaterally across nodes.

\subsection{What if we keep the outliers?}\label{sec:outliers_keep}
Figure \ref{fig:basic_interp_plot} vividly illustrates the issue. With linear interpolation at cap-maturity nodes, the caplet volatility curve exhibits large, highly unstable oscillations, diving deeply negative around the problematic 18-month maturity mark.

\begin{figure}[ht]
\centering
\includegraphics[width=0.95\textwidth]{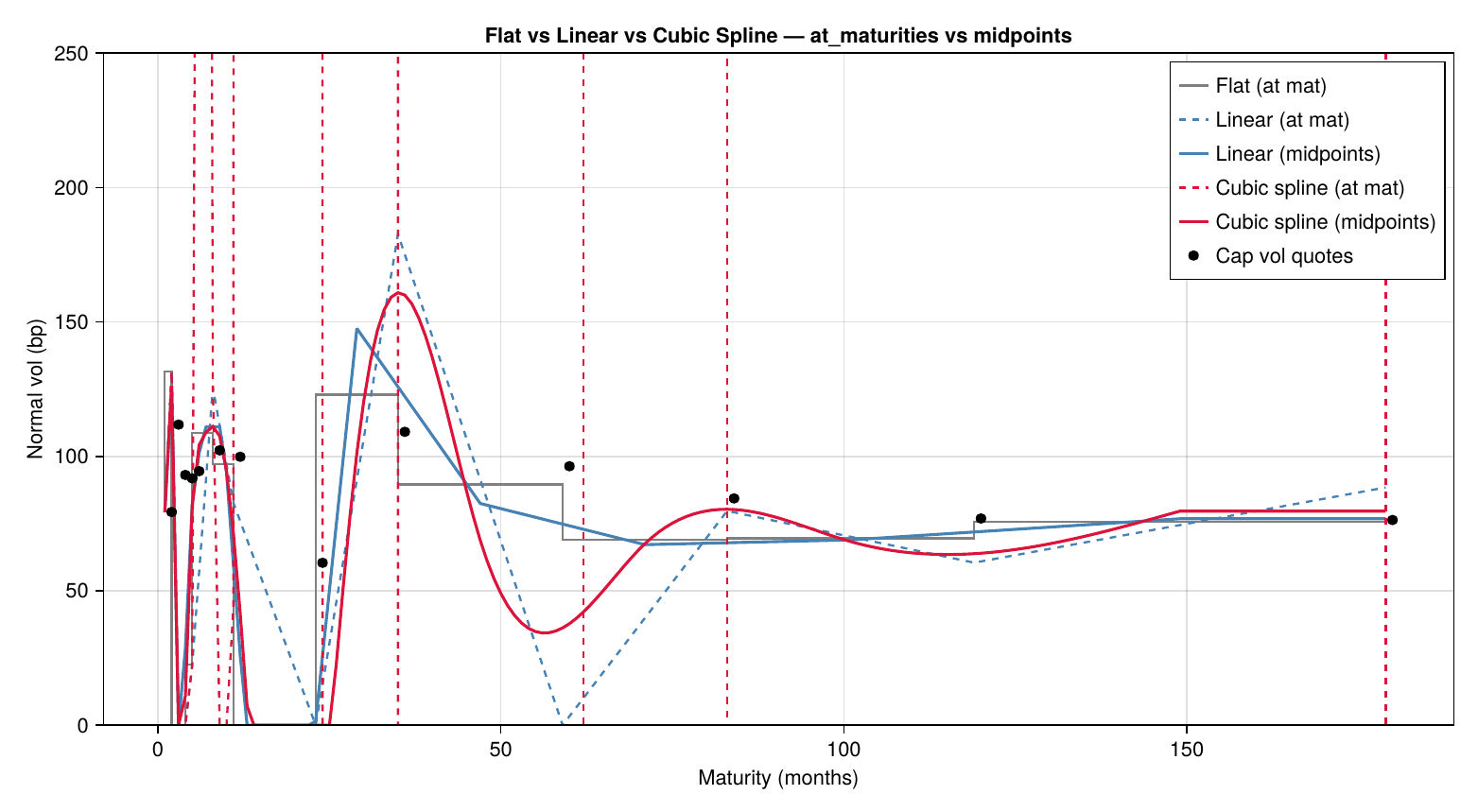}
\caption{Caplet volatilities obtained with linear interpolation at maturity nodes exhibit strong oscillations. \label{fig:basic_interp_plot}}
\end{figure}

\subsection{Global minimization with positivity constraint}
If we do not remove outliers, the bootstrap or global solver will try to fit the bad quotes, thus producing large oscillations and negative vols. Enforcing positivity with an exponential transform or a non-negative filter will help to mitigate the issue.
Figure \ref{fig:hyman_comparison_plot} shows that both approaches work. The $C^1$ cubic spline with non-negative filter does not present any oscillation and is smoother than linear interpolation with exponential transform.

\begin{figure}[ht]
\centering
\includegraphics[width=0.95\textwidth]{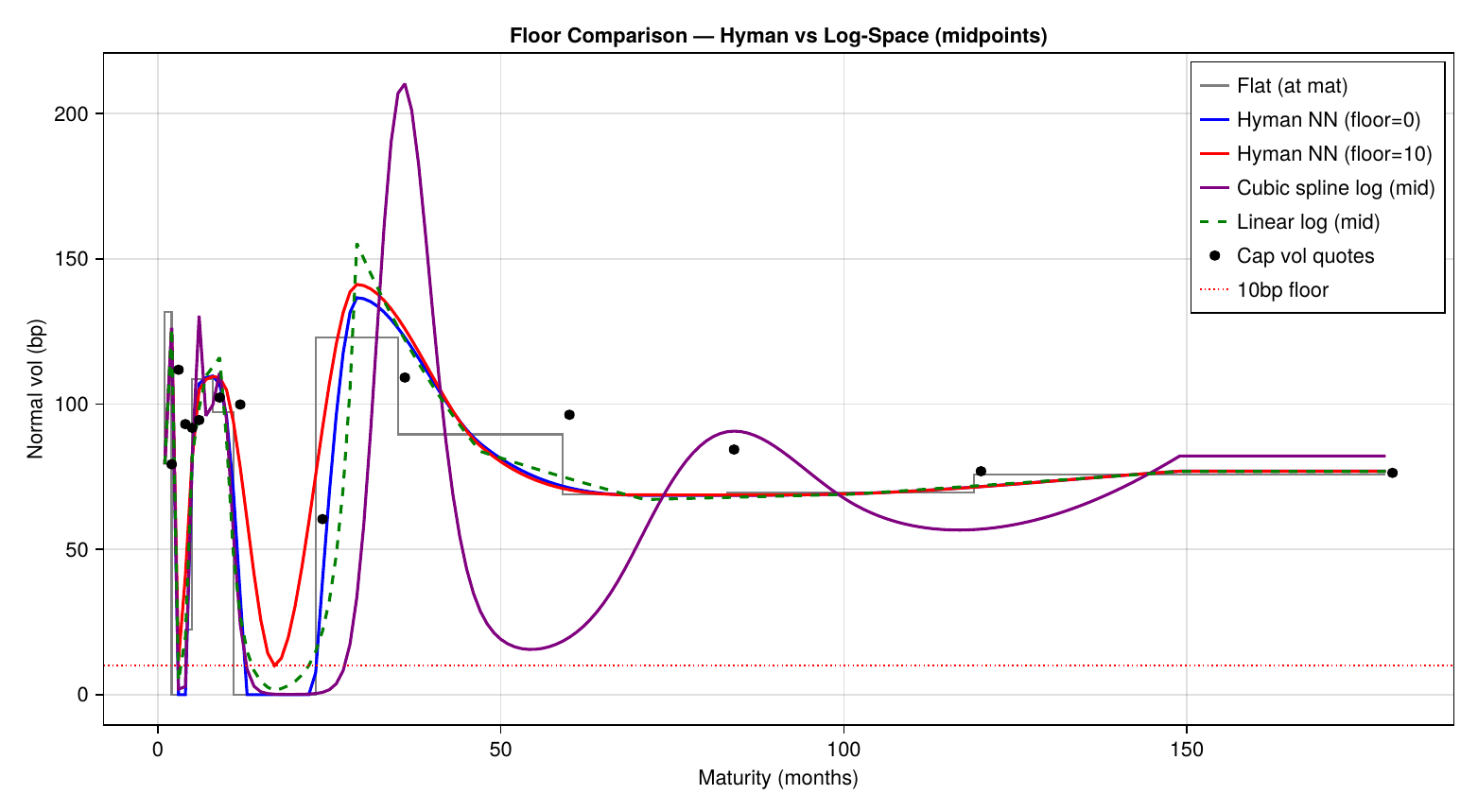}
\caption{Caplet volatilities obtained with non-negative constraint and midpoints nodes. \label{fig:hyman_comparison_plot}}
\end{figure}

\begin{table}[ht]
\caption{Cap repricing error for different interpolation methods.}
\centering
\begin{tabular}{l r r r}
\toprule
Method & Min vol (bp) & Min node & Reprice err \\
\midrule
flat at maturity & 0.0 & 0.0 & 5.04e-03\\
linear at maturity & 0.0 & 0.0 & 8.61e-03\\
cubic at maturity & 0.0 & 0.0 & 5.78e+05\\
linear mid & 0.0 & -144.56 & 3.70e-03\\
cubic mid & 0.0 & -112.02 & 3.70e-03\\
hyman mid & 0.0 & -121.46 & 3.70e-03\\
hyman mid floor=10 (shift) & 10.0 & 10.0 & 5.75e-03\\
linear exp mid & 1.42 & 1.42 & 3.70e-03\\
cubic exp mid & 0.09 & 0.18 & 3.70e-03\\
\bottomrule
\end{tabular}
\label{tbl:method_comparison}
\end{table}
\section{Market Data}\label{app:market_data}

\tablesize{\scriptsize}
\begin{table}[H]
\caption{Cap volatilities (in basis points) on Libor-1M as of February 2022 for a strike $K=0\%$. \label{tbl:cap_vols}}
\begin{tabular}{c|ccccccccccccc}
Maturity (months) & 2 & 3 & 4 & 5 & 6 & 9 & 12 & 24 & 36 & 60 & 84 & 120 & 180 \\
\hline
Cap Vol (bp) & 79.30 & 111.83 & 93.10 & 91.83 & 94.51 & 102.25 & 99.86 & 60.42 & 109.15 & 96.34 & 84.37 & 76.90 & 76.34 \\
\end{tabular}
\end{table}
\begin{table}[H]
\caption{USD OIS discount curve as of February 2022. \label{tbl:discount_curve}}
\begin{tabular}{c|cccccccccccccc}
Maturity (months) & 0 & 1 & 2 & 3 & 6 & 12 & 24 & 36 & 60 & 84 & 120 & 180 & 240 & 360 \\
\hline
Zero Rate (\%) & 0.05 & 0.06 & 0.20 & 0.35 & 0.69 & 1.01 & 1.44 & 1.62 & 1.71 & 1.81 & 1.83 & 2.04 & 2.25 & 2.17 \\
\end{tabular}
\end{table}
\begin{table}[H]
\caption{Libor 1M zero rates as of February 2022. \label{tbl:zero_curve}}
\begin{tabular}{c|cccccccccccccc}
Maturity (months) & 0 & 1 & 2 & 3 & 6 & 12 & 24 & 36 & 60 & 84 & 120 & 180 & 240 & 360 \\
\hline
Zero Rate (\%) &0.16 & 0.16 & 0.3 & 0.46 & 0.8 & 1.12 & 1.34 & 1.55 & 1.73 & 1.82 & 1.92 & 1.94 & 2.15 & 2.28\\
\end{tabular}
\end{table}

\section{Interest Curve Interpolation}\label{sec:forward_curve}
In the cap price or time-value interpolation approach, the interpolation of the discount curve and forward curve is also important. In particular, if we use a linear interpolation in the logarithm of discount factors (equivalently, a piecewise-constant interpolation for the forward curve), we will have a jump in the forward rate at the caplet maturities (Figure \ref{fig:forward_rate_comparison}). 
\begin{figure}[H]
\centering
\includegraphics[width=0.95\textwidth]{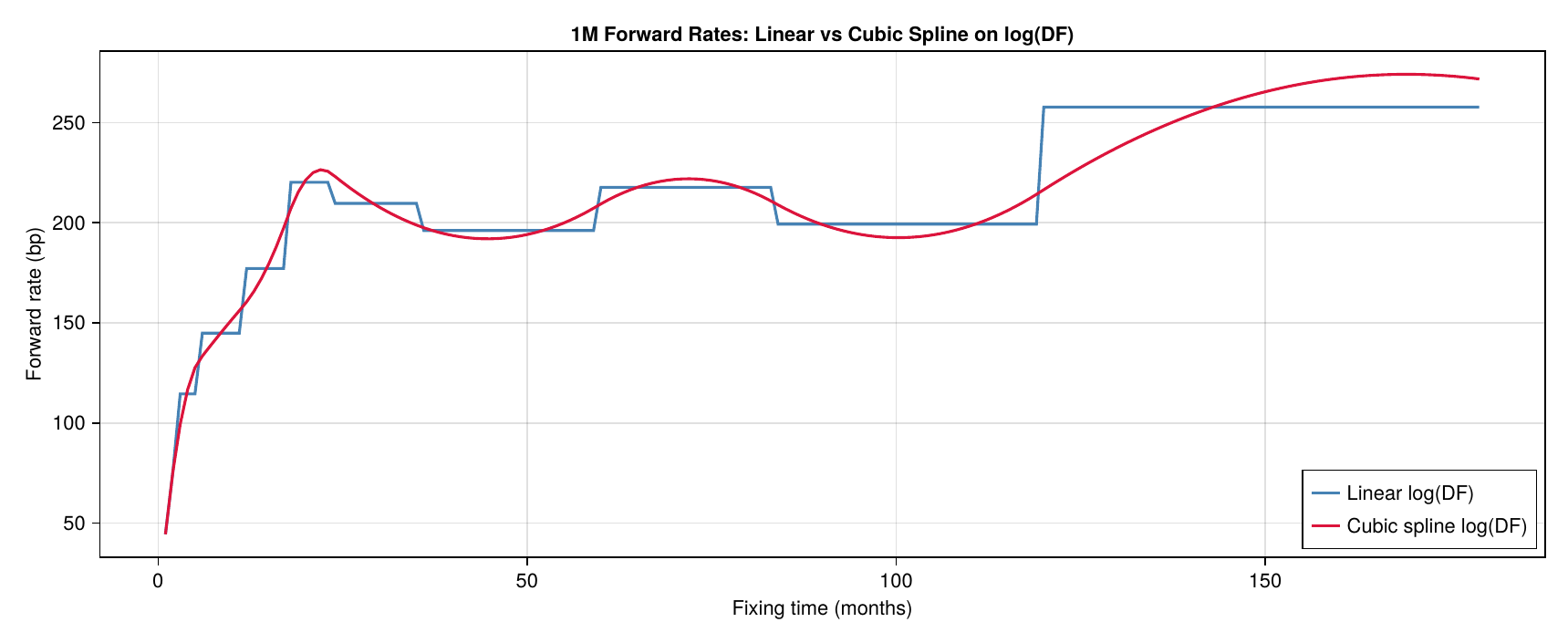}
\caption{Forward rates using linear interpolation  or a $C^2$ cubic spline on the logarithm of discount factors. \label{fig:forward_rate_comparison}}
\end{figure}
This will lead to jumps in the caplet volatilities even if the interpolator used for cap time-values is smooth (Figure \ref{fig:c2hyman_lin_vs_c2_clean}).

\begin{figure}[H]
\centering
\includegraphics[width=0.95\textwidth]{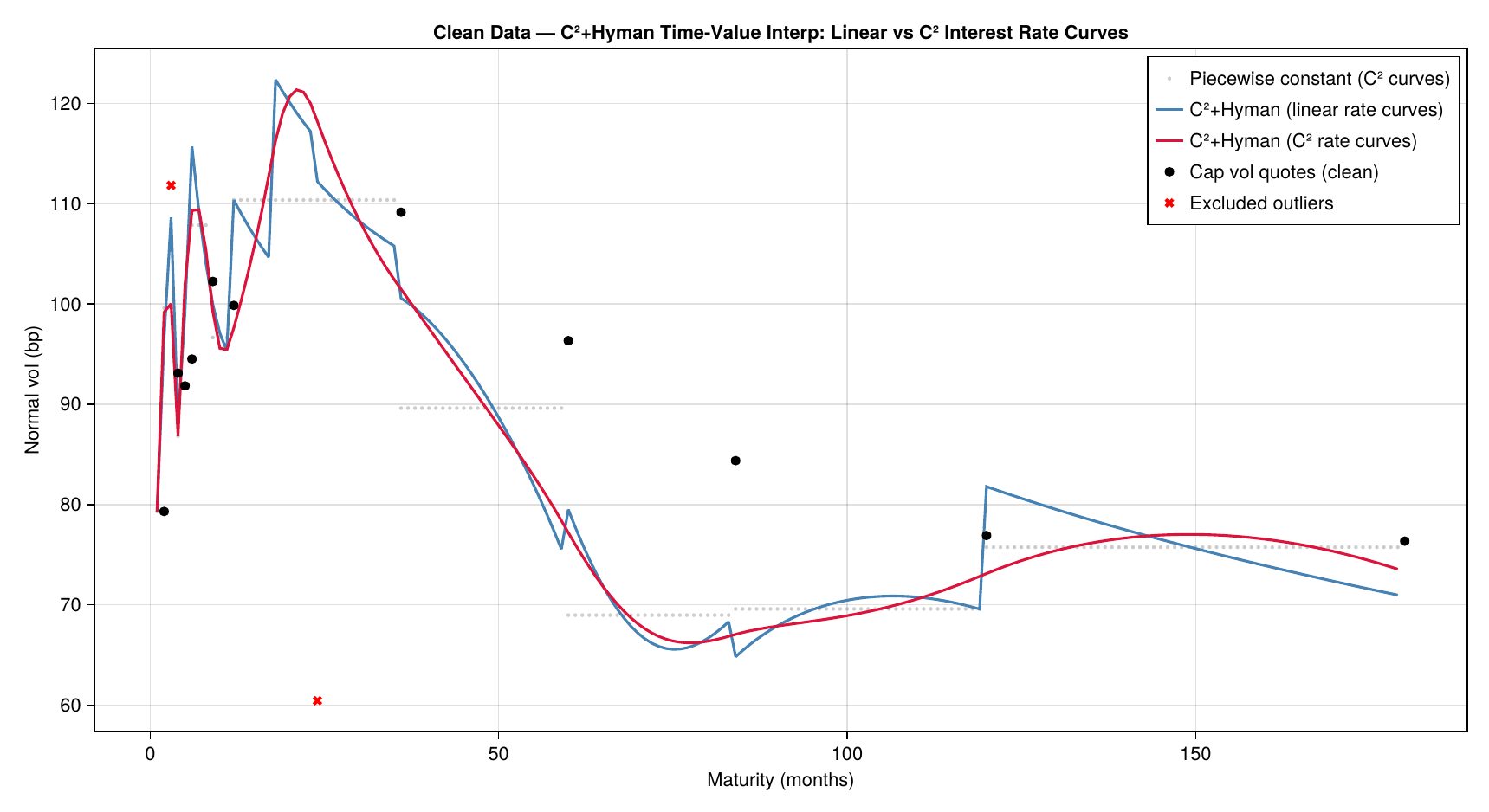}
\caption{Interest rate curve interpolation impact on the caplet volatilities. \label{fig:c2hyman_lin_vs_c2_clean}}
\end{figure}

\section{Stripping Without Boostrapping: Extrapolation}\label{sec:extrapolation}
Figure \ref{fig:extrap_usd_K0} shows the impact of the choice of fake cap quote date on the interpolation up to the last quoted tenor, for our example of 1M Libor cap volatilities. 
\begin{figure}[H]
\centering
\includegraphics[width=0.95\textwidth]{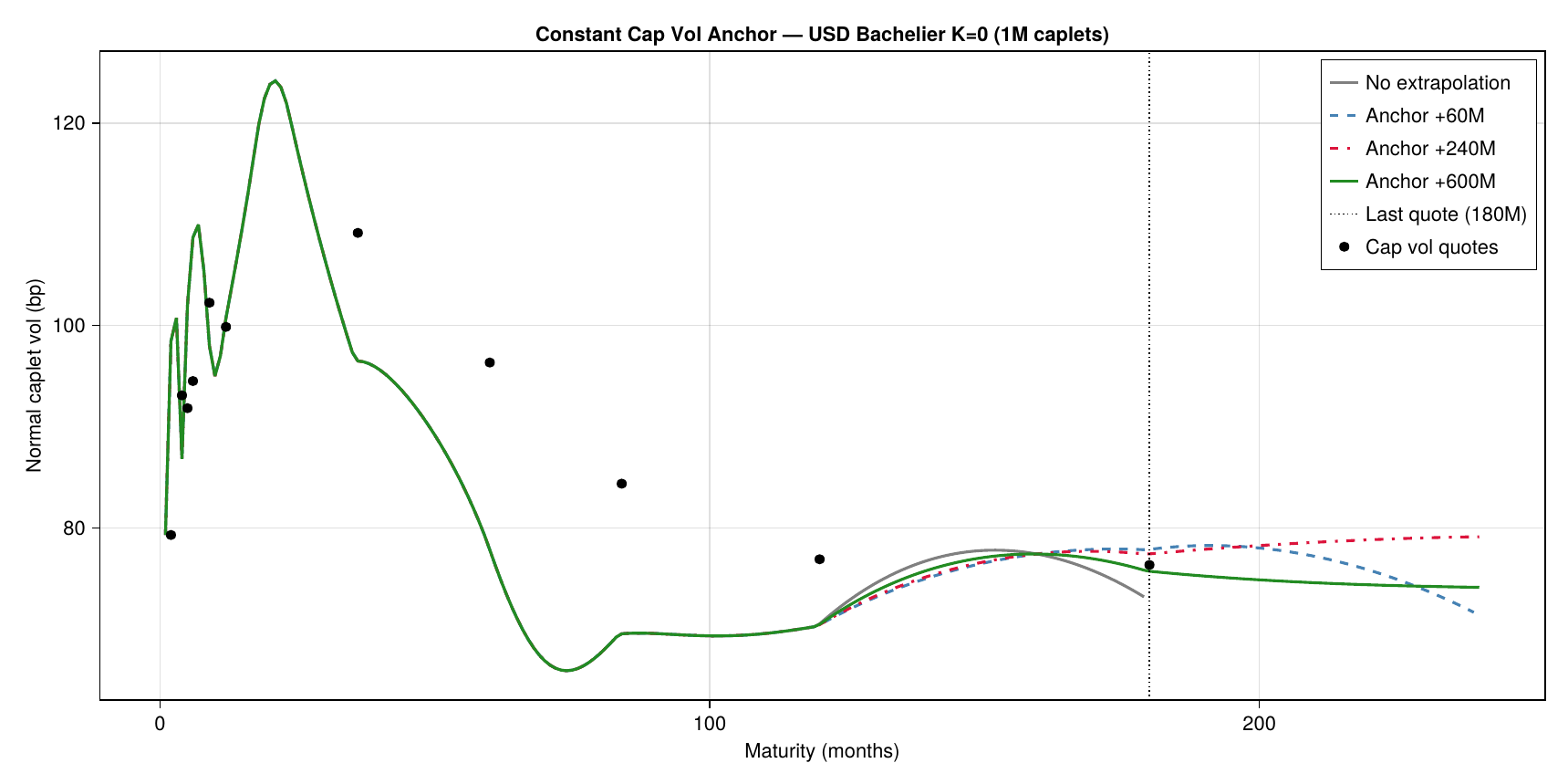}
\caption{Caplet vols adding a fake cap quotes far away (at different dates), with volatility equal to the last quoted maturity. \label{fig:extrap_usd_K0}}
\end{figure}

\section{Detailed Cap Price Tables}

\begin{table}[H]
  \caption{Cap prices and time values at $K=0$~bp.}
\label{tab:tv_K0}
\centering
\begin{tabular}{rr rrrr}
\toprule
$T$ (M) & $\bar\sigma$ (bp) & Cap Price & Intrinsic & Time Value \\
\midrule
  2 &  79.30 &     3.7266 &     3.7078 &   0.0188 \\
  3 & 111.83 &    10.1769 &     9.9959 &   0.1810 \\
  4 &  93.10 &    19.6077 &    19.5256 &   0.0821 \\
  5 &  91.83 &    29.1470 &    29.0471 &   0.0998 \\
  6 &  94.51 &    38.7387 &    38.5605 &   0.1782 \\
  9 & 102.25 &    75.1026 &    74.5777 &   0.5249 \\
 12 &  99.86 &   111.4764 &   110.4753 &   1.0011 \\
 24 &  60.42 &   305.1656 &   305.0634 &   0.1022 \\
 36 & 109.15 &   520.6022 &   506.5736 &  14.0286 \\
 60 &  96.34 &   909.9606 &   873.0928 &  36.8678 \\
 84 &  84.37 &  1315.1064 &  1264.2562 &  50.8503 \\
120 &  76.90 &  1871.2891 &  1775.9994 &  95.2898 \\
180 &  76.34 &  2972.7366 &  2784.7155 & 188.0211 \\
\bottomrule
\end{tabular}
\end{table}

\begin{table}[H]
  \caption{Incremental time values at $K=0$~bp.  Two violations (bold).}
\label{tab:dtv_K0}
\centering
\begin{tabular}{rr rrrr}
\toprule
$T$ (M) & $\bar\sigma$ (bp) & $\Delta P$ & $\Delta\IV$ & $\Delta\TV$ \\
\midrule
  2 &  79.30 &     3.7266 &     3.7078 &   $+$0.0188 \\
  3 & 111.83 &     6.4503 &     6.2881 &   $+$0.1622 \\
  4 &  93.10 &     9.4308 &     9.5297 &   $\mathbf{-0.0989}$ \\
  5 &  91.83 &     9.5392 &     9.5215 &   $+$0.0177 \\
  6 &  94.51 &     9.5917 &     9.5133 &   $+$0.0784 \\
  9 & 102.25 &    36.3639 &    36.0172 &   $+$0.3467 \\
 12 &  99.86 &    36.3738 &    35.8976 &   $+$0.4762 \\
 24 &  60.42 &   193.6892 &   194.5882 &   $\mathbf{-0.8990}$ \\
 36 & 109.15 &   215.4366 &   201.5101 &  $+$13.9264 \\
 60 &  96.34 &   389.3584 &   366.5192 &  $+$22.8392 \\
 84 &  84.37 &   405.1459 &   391.1634 &  $+$13.9825 \\
120 &  76.90 &   556.1827 &   511.7432 &  $+$44.4395 \\
180 &  76.34 &  1101.4475 &  1008.7162 &  $+$92.7313 \\
\bottomrule
\end{tabular}
\end{table}

\begin{table}[H]
  \caption{Cap prices and time values at $K=200$~bp.}
\label{tab:tv_K200}
\centering
\begin{tabular}{rr rrrr}
\toprule
$T$ (M) & $\bar\sigma$ (bp) & Cap Price & Intrinsic & Time Value \\
\midrule
  2 &  79.30 &     0.0000 &     0.0000 &    0.0000 \\
  3 & 111.83 &     0.0037 &     0.0000 &    0.0037 \\
  4 &  93.10 &     0.0510 &     0.0000 &    0.0510 \\
  5 &  91.83 &     0.1467 &     0.0000 &    0.1467 \\
  6 &  94.51 &     0.3552 &     0.0000 &    0.3552 \\
  9 & 102.25 &     3.2408 &     0.0000 &    3.2408 \\
 12 &  99.86 &     6.8309 &     0.0000 &    6.8309 \\
 24 &  60.42 &    30.7472 &     9.8271 &   20.9201 \\
 36 & 109.15 &   129.7639 &    19.0849 &  110.6791 \\
 60 &  96.34 &   252.9935 &    19.0849 &  233.9086 \\
 84 &  84.37 &   384.7332 &    50.7977 &  333.9356 \\
120 &  76.90 &   579.5457 &    50.7977 &  528.7480 \\
180 &  76.34 &  1117.2233 &   276.7270 &  840.4963 \\
\bottomrule
\end{tabular}
\end{table}

\begin{table}[H]
  \caption{Incremental time values at $K=200$~bp.  All $\Delta\TV \geq 0$.}
\label{tab:dtv_K200}
\centering
\begin{tabular}{rr rrrr}
\toprule
$T$ (M) & $\bar\sigma$ (bp) & $\Delta P$ & $\Delta\IV$ & $\Delta\TV$ \\
\midrule
  2 &  79.30 &     0.0000 &     0.0000 &    $+$0.0000 \\
  3 & 111.83 &     0.0037 &     0.0000 &    $+$0.0037 \\
  4 &  93.10 &     0.0473 &     0.0000 &    $+$0.0473 \\
  5 &  91.83 &     0.0958 &     0.0000 &    $+$0.0958 \\
  6 &  94.51 &     0.2085 &     0.0000 &    $+$0.2085 \\
  9 & 102.25 &     2.8855 &     0.0000 &    $+$2.8855 \\
 12 &  99.86 &     3.5901 &     0.0000 &    $+$3.5901 \\
 24 &  60.42 &    23.9164 &     9.8271 &   $+$14.0893 \\
 36 & 109.15 &    99.0167 &     9.2578 &   $+$89.7589 \\
 60 &  96.34 &   123.2296 &     0.0000 &  $+$123.2296 \\
 84 &  84.37 &   131.7397 &    31.7128 &  $+$100.0270 \\
120 &  76.90 &   194.8125 &     0.0000 &  $+$194.8125 \\
180 &  76.34 &   537.6776 &   225.9293 &  $+$311.7483 \\
\bottomrule
\end{tabular}
\end{table}

\section{Total Variance Check Algorithm}
Another ad-hoc criteria may be to compute the total variance of the cap $\hat\sigma_q^2 T_q$. In particular, we notice that it is not increasing around those two outliers. This does not mean that there is a calendar spread arbitrage opportunity, since caps of different maturities are not directly comparable as each caplet operates on a different forward rate. However, it is a strong indication that there may be an issue with the quotes, and that we should investigate further by looking at the outliers with the modified Z-score method.
   \begin{listing}[H]    
\caption{Total variance check for arbitrage detection in \emph{Julia} language.}\label{code:total_variance_check}
\begin{minted}[breaklines,escapeinside=||,mathescape=true, frame=lines, fontsize=\small, framesep=2mm]{julia}
cap_maturities_months = Float64.([2, 3, 4, 5, 6, 9, 12, 24, 36, 60, 84, 120, 180])
cap_vols_bp = [79.30, 111.83, 93.10, 91.83, 94.51, 102.25, 99.86, 60.42,
 109.15, 96.34, 84.37, 76.90, 76.34]

T = cap_maturities_months ./ 360
tv = (cap_vols_bp .^ 2) .* T
arb_violations = findall(diff(tv) .< 0)
println("\nTotal-variance arbitrage violations at indices: ", arb_violations)
\end{minted}
\end{listing}

\section{Outlier Detection Algorithm}
   \begin{listing}[H]    
\caption{Detect outliers in the cap volatilities using the MAD method in \emph{Julia} language.}\label{code:outlier}
\begin{minted}[breaklines,escapeinside=||,mathescape=true, frame=lines, fontsize=\small, framesep=2mm]{julia}
ref = [median(cap_vols_bp[max(1,i-2):min(end,i+2)]) for i in eachindex(cap_vols_bp)]
res = cap_vols_bp .- ref
mad_val = median(abs.(res .- median(res)))
mod_z = 0.6745 .* (res .- median(res)) ./ mad_val
outlier_idx = findall(abs.(mod_z) .> 3.0)
println("\nOutlier indices: ", outlier_idx)
\end{minted}
\end{listing}

% -------------------------------------------------------------------
% Appendix: Correction to the Hyman non-negative constraint (DEH)
% -------------------------------------------------------------------
\section{Correction to the Hyman Non-Negative Constraint}\label{app:hyman_typo}

\citet{dougherty1989nonnegativity} propose a practical filter on first derivatives of a cubic Hermite interpolant that guarantees non-negativity when all nodal function values are non-negative. A typographical error in the printed form of the non-negativity bounds was pointed out by \citet{chasethedevil2017hyman_typo}. The error concerns the placement of the adjacent interval lengths in the derivative bounds; the corrected inequalities are reproduced here for completeness and for reference when implementing positivity-preserving splines.

Let the nodes be $(x_k,f_k)$ for $k=1,\dots,n$ with $f_k\ge 0$ and denote the interval lengths by $h_k = x_{k+1}-x_k$ ($k=1,\dots,n-1$). Let $d_k$ be the cubic-Hermite derivative (i.e.\ $f'(x_k)$).

Before applying the positivity preserving filter it is customary to initialise the nodal derivatives either using $C^2$ continnuity conditions or with the parabolic (Bessel) estimate. We use the latter. Denote the secant slopes by
\[
s_k = \frac{f_{k+1}-f_k}{h_k},\qquad k=1,\dots,n-1.
\]
Then the Bessel (parabolic) estimate for interior nodes is
\begin{equation}\label{eq:bessel_appendix}
d_k^{\mathrm{Bessel}} = \frac{h_k\,s_{k-1} + h_{k-1}\,s_k}{h_{k-1} + h_k},\qquad k=2,\dots,n-1.
\end{equation}

Common endpoint choices are one-sided secant values $d_1=s_1$ and $d_n=s_{n-1}$, or other one-sided parabolic estimates. In our caplet-stripping implementation, the extrapolation is flat and we thus set the last derivative to zero (i.e.\ $d_n=0$) for maximum smoothness. The Bessel estimate is second-order accurate on smooth data and provides a balanced initial slope on non-uniform grids, after which the  clamp is applied to guarantee non-negativity.

A sufficient condition that the cubic Hermite interpolant is non-negative on the two intervals adjacent to $x_k$ is
\[
-\frac{3\,f_k}{h_k} \;\le\; d_k \;\le\; \frac{3\,f_k}{h_{k-1}},\qquad f_k>0,
\]
where the left-hand inequality follows from enforcing non-negativity on $[x_k,x_{k+1}]$ and the right-hand inequality from $[x_{k-1},x_k]$. When $f_k=0$ the derivative should be set to zero to avoid division by zero and to prevent artificial sign changes.
In algorithmic form the commonly used clamp (filter) can be written as:
\begin{verbatim}
for k in 1:n
  if f[k] <= 0
    d[k] = 0.0
    continue
  end
  if k > 1
    d[k] = min(d[k], 3.0 * f[k] / h[k-1])
  end
  if k < n
    d[k] = max(d[k], -3.0 * f[k] / h[k])
  end
end
\end{verbatim}

Using the corrected bounds prevents the filter from applying overly permissive or overly restrictive clamps that would otherwise allow negative undershoots or unnecessarily flatten the spline. We recommend testing the implementation on representative datasets and reporting any remaining negative excursions as data issues rather than as failures of the interpolator.

\end{document}